\documentclass[11pt,twoside]{article}
\usepackage{fancyhead}
\usepackage{psfrag}
\usepackage{epsfig}
\usepackage{graphicx}
\usepackage{amsmath}
\usepackage{amssymb}
\usepackage{amsthm}
\usepackage{subfigure}
\usepackage{cite}
\usepackage[latin1]{inputenc}

\newtheorem{lemma}{Lemma}
\newtheorem{theorem}{Theorem}

\begin{document}

\vspace*{5mm}

\noindent
\textbf{\LARGE On Perfect Codes in the Johnson Graph
\footnote{This work was supported in part by the Israel
Science Foundation (ISF), Jerusalem, Israel, under Grant 230/08.}
}
\thispagestyle{fancyplain} \setlength\partopsep {0pt} \flushbottom
\date{}

\vspace*{5mm}
\noindent
\textsc{Natalia Silberstein, Tuvi Etzion} \hfill \texttt{\{natalys,etzion\}@cs.technion.ac.il} \\
{\small Computer Science Dept., Technion - Israel Institute of Technology, Haifa, ISRAEL}

\medskip

\begin{center}
\parbox{11,8cm}{\footnotesize
\textbf{Abstract.} In this paper we consider the existence of nontrivial
perfect codes in the Johnson graph $J(n,w)$.
We present  combinatorial and number theory techniques to provide
necessary conditions for existence of such codes and reduce the range
of parameters in which $1$-perfect and $2$-perfect codes may exist.
}
\end{center}

\baselineskip=0.9\normalbaselineskip

\section{Introduction}

Codes which attain the sphere packing bound are called perfect codes.
The Hamming metric and the Johnson metric are the most important metrics in
coding theory on which
perfect codes are defined.
While for the Hamming space all perfect codes over finite fields
are known~\cite{MWS79}, in the Johnson space it was conjectured
by Delsarte in
1970's~\cite[p. 55]{Del73} that there are no nontrivial perfect codes.
The general nonexistence
proof still remains an open problem, although many attempts to solve the problem
were made, e.g.\cite{Ban77,Ham82,Roos83,Mar92,Etz96,ESc04,Etz06,Gor06}.

The \emph{Johnson space} $V_{w}^{n}$ consists of all $w$-subsets
of a fixed $n$-set $N=\left\{ 1,2,...,n\right\}$, for given positive
integers $n$ and $w$ such that $0\leq w\leq n$. With the Johnson space
we associate the \emph{Johnson graph} $J(n,w)$ with the vertex set
$V_{w}^{n}$, where two  $w$-subsets
are adjacent if and only if their intersection is of size $w-1$.
A code $C$ of such $w$-subsets is called an $e$-\emph{perfect code}
in $J(n,w)$ if the $e$-spheres with centers at the codewords of
$C$ form a partition of $V_{w}^{n}$. In other words, $C$ is an
$e$-perfect code if for each element $v\in V_{w}^{n}$ there exists
a unique codeword $c\in C$ such that the distance (in the graph) 
between $v$ and $c$ is at most $e$.

A code $C$ in $J(n,w)$ can be described as a collection of $w$-subsets
of $N$, but it can be also described as a binary code of length $n$
and constant weight $w$. From a $w$-subset $S$ we construct
a \emph{characteristic} binary
vector of length $n$ and weight $w$ with \textit{ones} in the positions of
$S$ and \textit{zeroes} in the positions of $N\setminus S$.
The \emph{Johnson distance}  between two $w$-subsets
is half of the number of coordinates in which their characteristic vectors differ.
In the sequel we will use a mixed language of sets and binary vectors.

There are three families of trivial perfect codes  in $J(n,w)$:
 $V_{w}^{n}$ is $0$-perfect;
any $\left\{ v\right\} $, $v\in V_{w}^{n}$, $w\leq n-w$, is $w$-perfect;
and if $n=2w$, $w$ odd, any pair of disjoint $w$-subsets is $e$-perfect
with $e=\frac{1}{2}(w-1)$.

In this paper we are interested in three problems concerning perfect codes in $J(n,w)$:
the existence of $1$-perfect codes, the existence of perfect codes in $J(2w,w)$, 
and improving the Roos bound~\cite{Roos83}
given by:
\begin{theorem}
\label{roos83} If an $e$-perfect code
in $J(n,w)$, $n\geq2w$, exists, then $$n\leq(w-1)\frac{2e+1}{e}.$$
\end{theorem}

The rest of this paper is organized as follows. In Section \ref{sec:1-perfect}
we present  new divisibility
conditions, which are based on the connection between perfect codes
in $J(n,w)$  and block designs. Based on this connection we 
introduce an improvement
of the Roos bound for $1$-perfect codes. In Section \ref{sec: 2-perfect}
we examine $2$-perfect codes in $J(2w,w)$
and present necessary conditions for existence of such codes, using Pell equation.
Finally in Section~\ref{sec:conclusion}  we
summarize our results.

\vspace{-0.2 cm}
\section{One-perfect codes in $J(n,w)$}
\label{sec:1-perfect}

In this section we consider first the  connection between  perfect codes in the
Johnson graph and block designs.

Let $t,n,w,\lambda$ be integers with
$n>w\geq t$ and $\lambda>0$, and let $N$ be an $n$-set. 
A $t-(n,w,\lambda)$ \emph{design}
is a collection $C$ of distinct $w$-subsets (called \emph{blocks}) of $N$
with the property that any $t$-subset of $N$ is contained in exactly
$\lambda$ blocks of $C$. Clearly, $C$ is a code in $J(n,w)$. On the other hand,
the largest $t$  for which a code  $C$ in $J(n,w)$ is a $t$-design
is called the \emph{strength} of the code.

In the sequel we assume w.l.o.g. that $n=2w+a$ for some $a\geq 0$, since the complement
of an $e$-perfect code in $J(n,w)$ is an $e$-perfect code in $J(n,n-w)$~\cite{Etz96}.

\begin{lemma}\cite{Etz06}
\label{sphere}
If the code $C$ in $J(2w+a,w)$ has strength $\varphi$ then for
each $t$, $0 \leq t\leq \varphi$, it
is a $t$-design $t-(2w+a,w,\lambda_{t})$ with
\begin{equation*}
\lambda_{t}=\frac{\binom{2w+a-t}{w-t}}{\Phi_{e}(w,a)},
\end{equation*}
where $\Phi_{e}(w,a)$ is the size of the sphere with radius $e$.
\end{lemma}

The strength of
a possible $e$-perfect code $C$ in $J(2w+a,w)$ can be used to exclude its existence.
We define the polynomial
\begin{equation}
\sigma_{e}(w,a,t)=\sum_{i=0}^{e}(-1)^{i}\binom{t}{i}
\sum_{j=0}^{e-i}\binom{w-i}{j}
\binom{w+a-t+i}{i+j}.\label{poly}
\end{equation}
It was proved in \cite{ESc04} that  $C$ is an $e$-perfect
code in $J(2w+a,w)$ with strength $\varphi$ if $\varphi$ is the
smallest positive integer  for which $\sigma_{e}(w,a,\varphi+1)=0$.

The  expression for the strength of an $1$-perfect
code can be easily calculated  from (\ref{poly}). We use this expression
to obtain  divisibility conditions for such codes and
to improve the Roos bound on the length of $1$-perfect codes in the Johnson graph.

\begin{theorem}
\label{thm:divisibility}
Assume  there exists an $1$-perfect code
$C$ in $J(2w+a,w)$ with strength $\varphi=w-d$ for some $d\geq0$. Then $d>1$,
$d\equiv0\mbox{ or }1(\mbox{mod }3)$, $w-d\equiv0,1,4,\mbox{ or }9(\mbox{mod }12)$, and
\begin{equation}\frac{\prod_{i=0}^{d-2}(wd-(d+i(d-1)))}{(d-1)!(d-1)^{d-1}d(w-d+1)}
    \in\mathbb{Z}.\label{divis}\end{equation}
\end{theorem}

\begin{proof} Assume that there exists an $1$-perfect code $C$
in $J(2w+a,w)$. Therefore, by (\ref{poly}), the strength of $C$
is $w-d=\frac{2w+a-1-\sqrt{(a+1)^{2}+4(w-1)}}{2}$.
Hence it follows that $d>1$ and
\begin{equation}a=\frac{w-d^{2}+d-1}{d-1}.\label{a(w,d)}
\end{equation}

In \cite{ESc04} it was proved that if there exists an $1$-perfect code in
$J(n,w)$ then either $w\equiv n-w\equiv1(\mbox{mod }12)$ or
$w\equiv n-w\equiv7(\mbox{mod }12)$.
In particular, $w\equiv1(\mbox{mod }6)$,   $12$ divides $a$,
and hence by (\ref{a(w,d)}), $6$ divides $d^{2}-d$.
Therefore, $d\equiv0\mbox{ or }1(\mbox{mod }3)$, i.e.,
$d\equiv0,1,3,4,6,7,9\mbox{ or }10(\mbox{mod }12).$ Since $12$ divides $a$,
from (\ref{a(w,d)}) it follows that $w\equiv d^{2}-d+1(\mbox{mod }12)$,
i.e., $w-d\equiv (d-1)^{2}(\mbox{mod }12)$.
Therefore,  $w-d\equiv0,1,4,\mbox{ or }9(\mbox{mod }12)$.

Now, using (\ref{a(w,d)}) we write the size of the sphere with radius $1$ as follows:
$\Phi_{1}(w,a)=1+w(w+a)=(w+a+d)(w-d+1)$.
By Lemma \ref{sphere}, for $t=w-d$ we have
$\lambda_{w-d}=\frac{\binom{w+a+d}{w+a}}{(w+a+d)(w-d+1)}
=\frac{\binom{w+a+d-1}{d-1}}{d(w-d+1)}\in\mathbb{Z}$.
We simplify the last expression, by using  $w+a+d-1=\frac{wd-d}{d-1}$
(implied by (\ref{a(w,d)})),
and  obtain the  divisibility condition~(\ref{divis}).
\end{proof}

By Theorem \ref{roos83} it follows that if an $1$-perfect code exists
in $J(2w+a,w)$, then $2w+a\leq3(w-1)$ and thus $a\leq w-3$.
We use the divisibility condition~(\ref{divis}) in Theorem~\ref{thm:divisibility} in
order to improve this bound.

\begin{theorem}
\label{Roos} If an $1$-perfect code exists in $J(2w+a,w)$,
then $a<\frac{w}{11}$.
\end{theorem}

\begin{proof}
Assume that there exists an $1$-perfect code
$C$ with strength $w-d$ in $J(2w+a,w)$.
We examine the divisibility condition (\ref{divis}) for  several
values of $d$ which are not pruned out by  Theorem \ref{thm:divisibility}.
\begin{itemize}
\item
Assume $d=3$. By (\ref{divis}) we have that $\frac{(3w-3)(3w-5)}{2!2^{2}3(w-2)}
=\frac{(w-1)(3w-5)}{8(w-2)}\in\mathbb{Z}$, which is impossible
since $\mbox{g.c.d.}(w-1,w-2)=1$ and $\mbox{g.c.d.}(3w-5,w-2)=1$.
Hence, $d>3$.
\item
Assume $d=4.$ By (\ref{divis}) we have that $\frac{4(w-1)(4w-7)2(2w-5)}{3!3^{3}4(w-3)}
\in\mathbb{Z}$.
Since $\mbox{g.c.d.}(w-3, w-1)\in \{1, 2\}$, $\mbox{g.c.d.}(w-3, 4w-7)\in\{1, 5\}$,
 and $\mbox{g.c.d.}(w-3, 2w-5)=1$,
it follows that
$w-3\leq 2\cdot5$. But by (\ref{a(w,d)}) we have $a=\frac{w-13}{3}$, and hence
$w>13$. Thus, $d>4$.
\item
Assume $d=6$. By (\ref{divis}) we have that
$\frac{6(w-1)(6w-11)(6w-16)(6w-21)(6w-26)}{5!5^{5}6(w-5)}\in\mathbb{Z}$.
Examining the $\mbox{g.c.d.}$ of $w-5$ with each factor of the nominator, we obtain 
that all possible factors of $w-5$ are from $\{2,3,7,19\}$. Since
$w\equiv d^{2}-d+1\equiv7(\mbox{mod }12)$, i.e., $w-5\equiv 2(\mbox{mod }12)$ we obtain that
$w-5=2\cdot7$, or $w-5=2\cdot19$, or $w-5=2\cdot7\cdot19$. Therefore $w=19$, 43
or 271. By (\ref{a(w,d)})  we have that $a=\frac{w-31}{5}$, and therefore the only 
possible value for $w$
is 271 and $a=48$.  But  $\Phi_{1}(w,a)$ does not divide $\binom{2w+a-(w-7)}{w+a}$ which
contradicts  Lemma \ref{sphere}.
Thus, $d>6$.

Similarly we obtain contradiction  for $d=7$ and $d=9$ and hence $d\geq 10$.
\item
Assume $d=10$. By (\ref{divis}) we have that

$\frac{10(w-1)(10w-19)(10w-28)(10w-37)(10w-46)
 (10w-55)(10w-64)(10w-73)(10w-82)}{9!*9^{9}*10(w-9)}\in\mathbb{Z}$.
Examining the $\mbox{g.c.d.}$ of $w-9$ with each factor of the nominator, we obtain that
all possible factors of $w-9$ are from $\{2,7,11,13,17,31,53,71\}$. By using that
$w-9\equiv -2(\mbox{mod }12)$ since $w\equiv d^{2}-d+1\equiv7(\mbox{mod }12)$,
and examining the divisibility condition
of Lemma \ref{sphere} we get the contradiction.

\end{itemize}

Hence, since $d\equiv0,1(\mbox{mod }3)$, it follows that $d\geq12$. Thus, 
by (\ref{a(w,d)}) we have
$$a\leq\frac{w-12^{2}+12-1}{11}=\frac{w-133}{11}<\frac{w}{11}.$$
\vspace{-0.2 cm}
\end{proof}
Clearly, as the value of $d$ is growing, considering the divisibility
condition~(\ref{divis}) becomes more complicated. But, the same method can be used for 
further improving the bound of Theorem \ref{Roos}.
\vspace{-0.2 cm}

\section{Two-perfect codes in $J(2w,w)$}
\label{sec: 2-perfect}
In this section we show the necessary conditions for the existence
of a $2$-perfect code in $J(2w,w)$ using Pell equation and prove that
there are no $2$-perfect codes in $J(2w,w)$ for $w\leq 1.97\times 10^{7655}$.

\begin{theorem}
\label{thm:pell}
If a $2$-perfect code $C$ exists in
$J(2w,w)$, then there is an integer $m\geq 0$  such that
\begin{description}
  \item[$(c.1)$] $w=\frac{(1+\sqrt{2})^{2m+1}+(1-\sqrt{2})^{2m+1}+6}{4}$, and
  \item[$(c.2)$] $\gamma \triangleq \sqrt{2}((1+\sqrt{2})^{2m}-(1-\sqrt{2})^{2m})+1$
is a square of some integer.
\end{description}
\end{theorem}

\begin{proof}

Assume $C$ is a $2$-perfect code in $J(2w,\; w)$. By (\ref{poly})
the strength of such code is
\begin{equation}\frac{1}{2}(-1+2w-\sqrt{8w-11\pm4\sqrt{5-6w+2w^{2}}}).\label{2-strength}
\end{equation}
Hence, there exists an integer $y$ such that
$y^2=5-6w+2w^{2}$ and thus
$(2w-3)^{2}-2y^{2} = -1$.
Let $x=2w-3$, and consider the equation $$x^{2}-2y^{2}=-1.$$
This equation is known as Pell equation~\cite{Pell} and it has a family of solutions given by
\begin{eqnarray}
x & = & \frac{(1+\sqrt{2})^{2m+1}+(1-\sqrt{2})^{2m+1}}{2}\label{eq:x}\\
y & = & \frac{(1+\sqrt{2})^{2m+1}-(1-\sqrt{2})^{2m+1}}{2\sqrt{2}}\label{eq:y}\end{eqnarray}
for some integer $m\geq 0$. Hence
$$w=\frac{(1+\sqrt{2})^{2m+1}+(1-\sqrt{2})^{2m+1}+6}{4},$$
which completes the proof of $(c.1)$.

Equation (\ref{2-strength}) implies also that
\[
\sqrt{8w-11\pm4\sqrt{5-6w+2w^{2}}}\in\mathbb{Z}.\]
We distinguish between two cases:
\begin{description}
  \item[Case 1:] $\sqrt{8w-11+4\sqrt{5-6w+2w^{2}}}\in\mathbb{Z}$.
In this case,  there exists an  integer $\alpha$ such that $\alpha^2=8w-11+4
\sqrt{5-6w+2w^{2}} = 8w-11+4y
= 4(x+y)+1$.
  \item[Case 2:] $\sqrt{8w-11-4\sqrt{5-6w+2w^{2}}}\in\mathbb{Z}.$
In this case,  there exists an  integer $\beta$ such that $\beta^2=
8w-11-4\sqrt{5-6w+2w^{2}} =8w-11-4y=4(x-y)+1$.
\end{description}
In other words, at least one of the expressions $4(x+y)+1$, $4(x-y)+1$  should be a
square of some integer. By (\ref{eq:x}) and (\ref{eq:y}) we obtain:
\begin{eqnarray*}
\alpha^2=4(x+y)+1 = \sqrt{2}((1+\sqrt{2})^{2m+2}-(1-\sqrt{2})^{2m+2})+1,\\
\beta^2=4(x-y)+1 = \sqrt{2}((1+\sqrt{2})^{2m}-(1-\sqrt{2})^{2m})+1,
\end{eqnarray*}
and the theorem is proved.
\end{proof}
We examine now  condition ($c.2$) of Theorem \ref{thm:pell} for 
$1\leq m\leq 10000$. The only values of $m$ which satisfy ($c.2$)
are $0,1$, and $2$, where $\gamma = 1,9$, and $49$, respectively.
The corresponding values of $w$ are $2,5$, and $22$,
respectively. For $w=2$ it  can be readily verified that there is no $2$-perfect
code in $J(4,2)$. It was proved in~\cite{ESc04} that if a $2$-perfect code exists in
$J(2w,w)$, then $w\equiv2,26,\mbox{ or }50(\mbox{mod }60)$. Hence there are no $2$-perfect
codes for $w=5$  and $w=22$. Thus for $1\leq w\leq 1.97\times 10^{7655}$ 
(considering $m=10000$),
there is no $2$-perfect code in $J(2w,w)$.
\vspace{-0.2 cm}

\section{Conclusion}
\label{sec:conclusion}

We have proved the nonexistence of perfect codes in $J(n,w)$ for large set of new parameters. 
Further techniques and some nonexistence results can be found in~\cite{Sil07}.

\end{document}